\documentclass[12pt]{article}
\usepackage{authblk}
\usepackage[utf8]{inputenc}
\usepackage[english]{babel}
\usepackage{amsthm}
\usepackage{amsmath}
\usepackage{amssymb}
\usepackage{color}
\usepackage{graphicx}
\usepackage{hyperref}
\usepackage[numbers,sort&compress]{natbib}
\usepackage{fullpage}
\usepackage{algorithm}
\usepackage[noend]{algpseudocode}
\floatname{algorithm}{Algorithm}

\newcommand{\Qed}{\hbox{ }\hfill\rule{2.5mm}{3mm}\medskip}

\def\RR{{\mathbb{R}}}

\theoremstyle{definition}
\newtheorem{theorem}{Theorem}[section]
\newtheorem{lemma}[theorem]{Lemma}

\newtheorem{definition}[theorem]{Definition}

\title{Finding long simple paths in a weighted digraph using pseudo-topological orderings}
\author{Miguel Raggi \thanks{Research supported in part by PAPIIT IA106316, UNAM. \\ Orcid ID: 0000-0001-9100-1655.} \\ \texttt{mraggi@gmail.com}}
\affil{Escuela Nacional de Estudios Superiores \\ Universidad Nacional Autónoma de México}
\date{}

\begin{document}

\maketitle

\begin{abstract}
Given a weighted digraph, finding the longest path without repeated vertices is well known to be NP-hard. Furthermore, even giving a reasonable (in a certain sense) approximation algorithm is known to be NP-hard. In this paper we describe an efficient heuristic algorithm for finding long simple paths, using an hybrid approach of heuristic depth-first search and pseudo-topological orders, which are a generalization of topological orders to non acyclic graphs, via a process we call ``opening edges''.
\end{abstract}

\textbf{Keywords:} long paths, graphs, graph algorithms, weighted directed graphs, long simple paths, heuristic algorithms.

\section{Introduction}

We focus on the following problem: Given a weighted digraph $D=(V,E)$ with weight $w:E\to \RR^+$, find a \emph{simple} path with high weight. The \emph{weight} of a path is the sum of the individual weights of the edges belonging to the path. A path is said to be \emph{simple} if it contains no repeated vertices.

Possible applications of this problem include motion planning, timing analysis in a VLSI circuit and DNA sequencing.

The problem of finding long paths in graphs is well known to be NP-hard, as it is trivially a generalization of HAMILTON PATH. Furthermore, it was proved by Björklund, Husfeldt and Khanna in \cite{bjorklund} that the longest path cannot be aproximated in polynomial time within $n^{1-\varepsilon}$ for any $\varepsilon > 0$ unless $P=NP$. 

While LONGEST SIMPLE PATH has been studied extensively in theory for simple graphs (for example in \cite{scutella2003approximation}, \cite{zhang2007algorithms}, \cite{Pham2012}), not many efficient heuristic algorithms exist even for simple undirected graphs, much less for weighted directed graphs. A nice survey from 1999 can be found at \cite{scholvin1999approximating}. A more recent comparison of 4 distinct genetic algorithms for approximating a long simple path can be found in \cite{portugal2010study}. 

An implementation of the proposed algorithm won the Oracle MDC coding competition in 2015. In the problem proposed by Oracle in the challenge ``Longest overlapping movie names'', one needed to find the largest concatenation of overlapping strings following certain rules, which could be easily transformed to a problem of finding the longest simple path in a directed weighted graph. The graph had around 13,300 vertices.

Our contribution lies in a novel method of improving existing paths, and an efficient implementation of said method. The proposed algorithm consists of two parts: finding good candidate paths using heuristic DFS and then improving upon those candidates by attempting to either replace some vertices in the path by longer subpaths--or simply insert some subpaths when possible--by using pseudo-topological orders.

The full C++ source code can be downloaded from 

\begin{center}
	\url{http://github.com/mraggi/LongestSimplePath}.
\end{center}

In Section~\ref{sec:prelim} we give some basic definitions. We describe the proposed algorithm in Section~\ref{sec:algo}. Finally, we give some implementation details and show the result of some experimental data in Section~\ref{sec:impl}.

It should be noted that for this particular problem it is generally easy to quickly construct somewhat long paths, but only up to a point. After this point even minor improvements get progressively harder.

\section{Preliminaries}\label{sec:prelim}
	
\begin{definition}
	A \emph{directed acylic graph} (or DAG) $D$ is a directed graph with no (directed) cycles.
\end{definition}

In a directed acyclic graph, one can define a partial order $\prec$ of the vertices, in which we say $v\prec u$ iff there is a directed path from $v$ to $u$.

\begin{definition}
	A \emph{topological ordering} for a directed acyclic graph $D$ is a total order of the vertices of $D$ that is consistent with the partial order described above. In other words, it is an ordering of the vertices such that there are no edges of $D$ which go from a ``high'' vertex to a ``low'' vertex.
\end{definition}

\begin{figure}[ht]
	\begin{center}
		\includegraphics[scale=1]{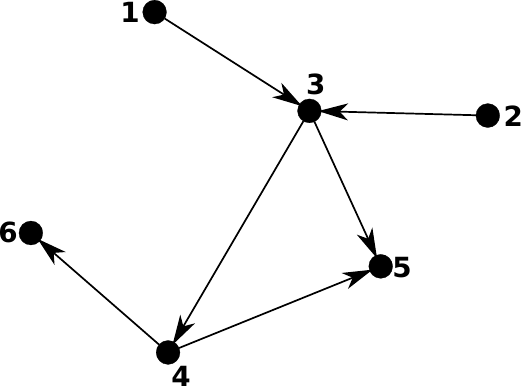}$\qquad$ \includegraphics[scale=1]{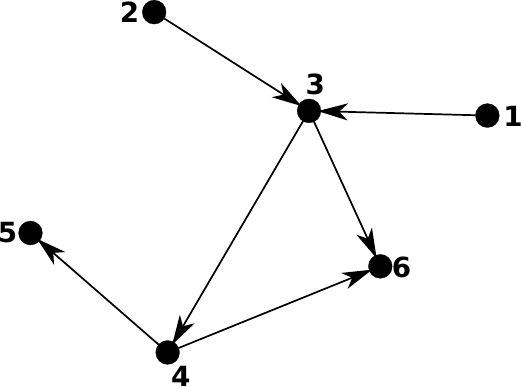}
	\end{center}
\caption{Two different topological orders of the same digraph}
\end{figure}

\begin{definition} 
	Given a digraph $D$, a \emph{strongly connected component} $C$ is a maximal set of vertices with the following property: for each pair of vertices $x,y \in C$, there exists a (directed) path from $x$ to $y$ and one from $y$ to $x$. A \emph{weakly connected component} is a connected component of the associated simple graph. 
\end{definition}

\begin{definition}
	Given a digraph $D$, the \emph{skeleton} $S$ of $D$ is the graph constructed as follows: The vertices of $S$ are the strongly connected components of $D$. Given $x,y \in V(S)$, there is an edge $x\to y$ iff there exists $a\in x$ and $b \in y$ such that $a\to b$ is an edge of $D$.
\end{definition}

It can be observed that $S$ is always a directed acyclic graph. 

\begin{definition}
Denote by $\overline{v}$ the connected component of $D$ which contains $v$. Given a vertex $v$, we define the \emph{out-rank} of $v$ as the length of the longest path of $S$ that \emph{starts} at $\overline{v}$. Similarly, we define the \emph{in-rank} of $v$ as the length of the longest path of $S$ which \emph{ends} at $\overline{v}$.
\end{definition}

\subsection{Longest simple path on DAGs}\label{subsection:lspdag}

In case that the digraph is a acyclic, a well-known algorithm that uses dynamic programming can find the optimal path in $O(n)$ time.

We describe Dijkstra's algorithm adapted to finding the longest simple path in a DAG. As this algorithm is an essential building block of the algorithm described in Section~\ref{pso}, we add a short description here for convenience. For a longer discussion see, for example, \cite{sedgewick2011algorithms}.

 \begin{enumerate}
	\item Associate to each vertex $v$ a real number $x[v]$, which will end up representing the weight of the longest simple path that \emph{ends} at $v$.
	\item Find a topological ordering of the vertices of $D$.
	\item In said topological order, iteratively set $x[v]$ to the max of $x[p]+w(p\to v)$ where $p\to v$ is an edge, or 0 otherwise.
	\item Once we have $x[v]$ for every vertex $v$, reconstruct the path by backtracking, starting from the vertex $v$ with the highest $x[v]$.
\end{enumerate}

\vspace{1em}

In more detail,

\begin{algorithm}[H]
\caption{Longest simple path in a DAG}\label{algo:lspdag}
\begin{algorithmic}
	\Require{A DAG $D$ with weight function $w:E(D)\to \RR$.}
	\Ensure{A longest simple path $P$ in $D$.}
	\Function{LSP\_DAG}{$D$}
	\State $x$ is an array of size $|V(D)|$, initialized with zeroes.
	\State Find a topological order $T$ for $V(D)$
	\For{$v \in T$}
		\State $x[v] := \text{max}\{x_p+w(p\to v)\ :\ p\to v\}$
	\EndFor
	\State $v := \text{argmax}(x)$
	\State $P := $ path with only $v$
	\State $T' := \text{reverse}(T)$
	\While{$x[v] \neq 0$}
		\State $u := $ an in-neighbor of $v$ for which $x[u] + w(u\to v) = x[v]$
		\State Add $u$ to the \textit{front} of $P$
		\State $v := u$
	\EndWhile
	\State
	\Return $P$
	\EndFunction
\end{algorithmic}
\end{algorithm}

This algorithm is simple to implement and efficient. Its running time is $O(E+V)$, where $E$ is the number of edges of $D$.

\section{The Algorithm}\label{sec:algo}

In what follows we shall assume we have preprocessed the graph and found the weakly connected components, the strongly connected components, and have fast access to both the outgoing edges and incoming edges for each vertex. As we may perform the following algorithm on each weakly connected component, without loss of generality assume $D$ is weakly connected.

Our proposed algorithm has two main parts: In the first part we find long paths using heuristic depth first search, choosing in a smart way which vertices to explore first, and in the second part we discuss a strategy to improve the paths obtained in the first part. Since the idea based on DFS is somewhat standard or straightforward, the main contribution of this paper lies in the ideas presented in the second part.

\subsection{Depth-first search}\label{DFS}
We describe a variation on depth-first search (DFS). 

The standard way of implementing a depth-first search is to either use a stack (commonly refered to as the \emph{frontier} or \emph{fringe}) of vertices to store unexplored avenues, or to use recursive calls (effectively using the callstack in \textit{lieu} of the stack).

If the graph is not acyclic, DFS may get stuck on a cycle. The standard way of dealing with this problem, when one simply wishes to see every vertex (and not every simple path, as in our case), is to store previously explored vertices in a data structure that allows us to quickly check if a vertex has been explored before or not.

However, for the problem of finding the longest simple path, it's not enough to simply ignore previously explored vertices, as we may need to explore the same vertex many times, as we may arrive at the same vertex from many different paths, and these need to be considered separately. Thus, for this problem, it is not possible to backtrack to reconstruct the path, as in many other problems.

This could be solved simply by modifying DFS slightly: make the \emph{frontier} data structure containing paths instead of only vertices. However, storing paths uses a large amount of memory and all the extra allocations might slow down the search considerably.

We propose a faster approach that results from only modifying a single path inplace. This is very likely not an original idea, but an extensive search in the literature did not reveal any article that considers depth-first search in this manner. Probably because for most problems the recursive or stack implementations are quite efficient, as they only need to deal with stacks of vertices and not stacks of paths. 

In this approach, instead of maintaining a stack of unexplored avenues, assume for each vertex the outgoing edges are sorted in a predictable manner. Later we will sort the outgoing edges in a way that explores the vertices with high probability of producing long paths first, but for now just assume any order that we know. Since this approach modifies the path in place, always make a copy of the best path found so far before any operation that might destroy the path.

Furthermore, assume we have a function $\textsc{NextUnexploredEdge}(P,\{u,v\})$ that takes as input a path $P$ and an edge $\{u,v\}$, in which the last vertex of $P$ is $u$, and returns the next edge $\{u,w\}$ in the order mentioned above for which $w \notin P$. This can be found using binary search, or even adding extra information to the edge, so that each edge remembers its index in the adjacency list of the first vertex. If there is no such edge, the function should return $null$. If no parameter $\{u,v\}$ is provided, it should return the first edge $\{u,w\}$ for which $w\notin P$.

We will construct a single path $P$ and modify it repeatedly by adding vertices to the back of $P$.

\begin{algorithm}[H]
\caption{Next Path in a DFS manner}\label{algo:nextpath}
\begin{algorithmic}[0]
	\Require{A weighted digraph $D$ and a path $P$, which will be modified in place}
	\Ensure{Either $done$ or $not\_done$}
	\Function{NextPath}{$P$}
		\State $last :=$ last vertex of $P$
		\State $t := \Call{NextUnexploredEdge}{P}$
		\While{$t = null$ and $|P| > 1$}
			\State $last := $ last vertex of $P$
			\State Remove $last$ from $P$
			\State $newLast := $ last vertex of $P$
			\State $t := \Call{NextUnexploredEdge}{P,\{newLast,last\}}$
		\EndWhile{}
		\If{$t = null$}
			\State{\textbf{return} $done$}
		\EndIf{}
		\State Add $t$ to the back of $P$
		\State{\textbf{return} $not\_done$}
	\EndFunction
\end{algorithmic}
\end{algorithm}

By repeatedly applying this procedure we can explore every path that starts at a given vertex in an efficient manner, but there are still too many possible paths to explore, so we must call off the search after a specified amount of time, or perhaps after a specified amount of time has passed without any improvements on the best so far.

Finally, we can do both forward and backward search with a minor modification to this procedure. So once a path cannot be extended forward any more, we can check if it can be extended backward. We found experimentally that erasing the the first few edges of the path before starting the backward search often produces better results.

\subsubsection{Choosing the next vertex}\label{subsec:choosing}

We give the details for efficiently searching forward, as searching backward is analogous.

So we are left with the following two questions: At which vertex do we start our search at? And then, while performing DFS, which vertices do we explore first? That is, how do we find a ``good'' order of the outgoing edges, so that good paths have a higher chance of being found quickly?

The first question is easily answered: start at vertices with high out-rank.

To answer the second question, we use a variety of information we collect on each vertex before starting the search:

\begin{enumerate}
	\item The out-rank and in-rank.
	\item The (weighted) out-degree and (weighted) in-degree.
	\item A \emph{score} described below.
\end{enumerate}

Once we find the score of each vertex, order the out-neighbors by rank first and score second, with some exceptions we will mention below. The score should not depend on any path, only on local information about the vertex, and should be fast to calculate.

Formally, let $k$ be a constant (small) positive integer. For each vertex $v$, let $A_k(v)$ be the sum of weights of all paths of length $k$ starting at $v$. For example, $A_1(v)$ is the weighted out-degree.

Given a choice of parameters $a_1, a_2, ..., a_k \in \RR^+$, construct the (out) score for vertex $v$ as
	$$\text{score}_{out}(v) = \sum_{i=1}^k a_iA_i(v)$$

Intuitively, the score of each vertex tries to heuristically capture the number (and quality) of paths starting at that vertex. High score means more paths start at a vertex.
	
When performing forward search, perhaps counter-intuitively, giving high priority to vertices with low score (as long as it is not 0) consistently finds better paths than giving high priority to vertices with high score. The reason for this is that exploring vertices with low score first means \emph{saving} the good vertices--those with high scores--for later use, once more restrictions are in place. Low scoring vertices are usually quickly discarded if there is no out, and so by leaving vertices with high degree for later, when the path is longer and so there are more restrictions about which vertices can be used, makes sense. An exception is if a vertex has degree 0. In this case, we give the vertex a low priority, as no further paths are possible.

Another exception is to give higher priority to vertices with very low indegree (for example, indegree 1), since if they are not explored in a path when first finding their parents, they will never be used again later in the path.

In addition, we also use the in-degree information in an analogous way.

\subsection{Pseudo-topological order}\label{pso}

The idea behind the second part of the algorithm is to try to improve paths by either inserting some short paths in-between or replacing some vertices by some short paths in an efficient way that covers both.

We begin by introducing some definitions.

\begin{definition}
	Given a digraph $D$, a \emph{weak pseudo-topological ordering} $\prec$ of the vertices of $D$ is a total order of the vertices in which whenever $x\prec y$ and there is an edge $y\to x$, then $x$ and $y$ are in the same strongly connected component.
\end{definition}

In other words, a weak pseudo-topological order is a total order that is consistent with the partial order given by the skeleton.

\begin{definition}
	Given a digraph $D$, a \emph{strong pseudo-topological ordering} $\prec$ of the vertices of $D$ is a total order of the vertices in which whenever $x\prec y$ and there is an edge $y\to x$, every vertex in the interval $[x,y]$ is in the same strongly connected component.
\end{definition}

\begin{figure}[ht]
	\begin{center}
	\includegraphics[scale=1]{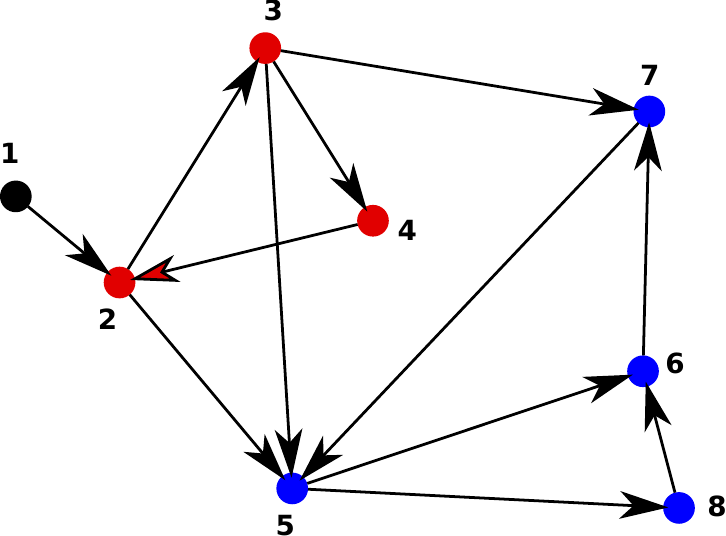}
\end{center}
\caption{A (strong) pseudo-topological ordering}
\end{figure}

In other words, a strongly pseudo-topological order is a weakly connected component in which the strongly connected components are not intermixed.

From here on, whenever we mention a pseudo-topological ordering, we mean a strong pseudo-topological ordering.

An easy way to get a random pseudo-topological ordering is to get a random topological ordering of the skeleton of the graph, and then ``explode'' the strongly connected components, choosing a random permutation of vertices in each component.

We can think of a pseudo-topological ordering as a topological ordering of the digraph in which we \emph{erase} all edges that go from a ``high'' vertex to a ``low'' vertex, thus considering an acyclic subdigraph of the original. We call this graph the \emph{subjacent DAG} of the pseudo-topological order. Thus, we may apply Algorithm~\ref{algo:lspdag} to this acyclic digraph and find its longest simple path.

As can be expected, the results obtained in this fashion are very poor compared to even non heuristic recursive slow depth-first search. However, if we combine the two approaches we get better paths.

\subsubsection{Combining the two approaches}

\begin{definition}
Given a path $P$ and a pseudo-topological ordering $T$, the \emph{imposition} of $P$ on $T$ is a pseudo-topological ordering $T_P$ which every vertex not in $P$ stays in the same position as in $T$. The vertices in $P$ are permuted to match the order given by $P$.
\end{definition}

For example, say we start with path $P = 3 \to 1 \to 5 \to 8$. Consider any pseudo-topological ordering, say, $T=(\textbf{1},\textbf{8},7,4,\textbf{3},6,\textbf{5},2)$. Then imposing the order defined by $P$ into $T$ gives rise to $T'=(\textbf{3},\textbf{1},7,4,\textbf{5},6,\textbf{8},2)$.

\begin{lemma}
$T_P$ as constructed above is also a (strong) pseudo-topological order. 
\end{lemma}
\begin{proof}
As $T$ is a strong pseudo-topological order, consider $S_1$, $S_2$, ... , $S_c$ the strongly connected components in the order they appear in $T$. Denote by $s(v)$ index of the strongly connected component of $v$. It suffices to prove that the vertices only move inside their own strongly connected components when going from $T$ to $T_P$.

\begin{figure} \label{fig:sccinpso}
\begin{center}
\includegraphics[scale=0.7]{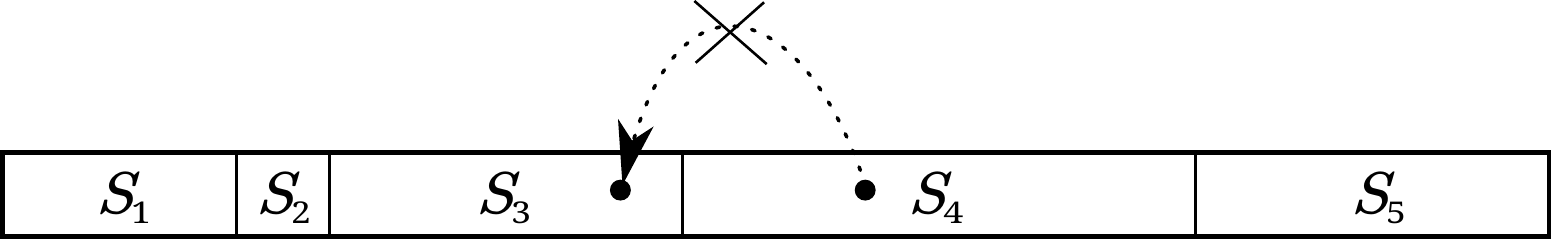}
\end{center}
\caption{No backward edges can jump between strongly connected components.}
\end{figure}

Let $(p_1,p_2,...,p_k)$ be the vertices of path $P$. Note that there is no $i < j$ for which $s(p_i) > s(p_j)$ since this would mean there exists a path from a vertex in $S_{s(p_i)}$ to a vertex in $S_{s(p_j)}$, but if $s(p_i) > s(p_j)$, no such path is possible in a pseudo-topological order, since this violates the order in the skeleton. This means that when imposing the order of $P$ into $T$ to get $T_P$, no vertex can jump out of their strongly connected component, and thus $T_P$ is also a strong pseudo-topological order.
\end{proof}

The previous lemma ensures that we may run algorithm \ref{algo:lspdag} with order $T_P$, and get a path that is \emph{at least as good} (and hopefully better) as path $P$, since all edges of $P$ remain in the subjacent DAG of $T_P$.

If after applying this technique we do find an improved path $P'$, we can repeat the process with $P'$, by again taking a random pseudo-topological ordering, imposing the order of $P'$ on this new ordering, and so on, until there is no more improvement.

The idea then is to construct long paths quickly with DFS and then use these paths as starting points for imposing on random pseudo-topological orders.

This approach does indeed end up producing moderately better paths than only doing DFS, even when starting from scratch with the trivial path and a random pseudo-topological order, albeit taking longer. However, we can do better.

\subsubsection{Opening the edges}

Again, we are in the setting where we have a path $P$ which we wish to improve.

Now, instead of just imposing path $P$ on multiple random pseudo-topological orders to find one that gets an improvement, construct orders as follows: Pick an edge $p_i\to p_{i+1}$ of path $P$ and construct a random pseudo-topological order that is consistent with $P$ and furthermore, for which $p_i$ and $p_{i+1}$ are as far apart as possible.

\begin{figure}[ht]
	\begin{center}
		\includegraphics[scale=0.8]{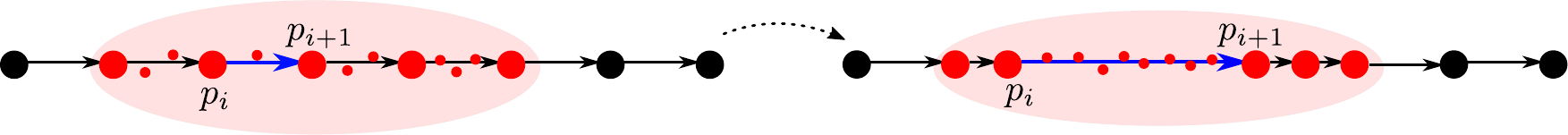}
	\end{center}
\caption{The process of opening an edge.}
\end{figure}

This is achieved by putting all vertices not in $P$ in all strongly connected components between $\overline{p_i}$ and $\overline{p_j}$ in between $p_i$ and $p_j$. In the figure above, the ``large'' vertices are vertices in $P$ and the ``small'' vertices are all other vertices not in $P$ in the same connected component as $p_i$ and $p_{i+1}$. If $p_{i}$ and $p_{i+1}$ are not in the same connected component, then, place \emph{every} vertex in either connected component, and also every vertex that belongs in a connected component between the component of $p_i$ and the component of $p_{i+1}$ between the two vertices, in such a way that the order is still a strong pseudo-topological order.

We may repeat this process for each edge in $P$.

The process of opening an edge is relatively expensive, since we must run Algorithm~\ref{algo:lspdag} each time. 

We now make an attempt at explaining why opening edges works as well as it does. Consider:
\begin{enumerate}
	\item If there exist a vertex $v$ that can be inserted into $P$, opening the corresponding edge finds this improvement.
	\item If there exists a vertex $v$ that can replace a vertex $p$ in $P$ and make the path longer (by means of edge weights), this process will find it when opening the edge to either side of $p$.
	\item Any (small) path of size $k$ that can be inserted into $P$, perhaps even replacing a few vertices of $P$, has probability at least $1/k!$ of being found if the corresponding vertices in the small path happen to be in the correct order.
\end{enumerate}

In the next section we try to heuristically maximize the probability that inserting or replacing paths will be found.

\subsubsection{Opening the edges eXtreme}

In the previous section, when opening up each edge, we put all the remaining vertices in the same connected component in a random order (consistent with pseudo-topological orders). We now consider the question of which order to give to those unused vertices. We discuss three different approaches: one heuristic that is quick and practical, one powerful and somewhat slow, and one purely theoretical using sequence covering arrays but which provides some theoretical guarantees.

Let $B$ be the inbetween vertices (\emph{i.e.} vertices between $p_i$ and $p_{i+1}$ when opening this edge). Since every other vertex will remain in their place, we face the problem of giving $B$ an ordering with a good chance of delivering an improvement. We only consider orders of $B$ that leaves the total order a strongly pseudo-topological order. That is, we only permute the vertices of $B$ within their own strongly connected components of the full graph.

Consider the induced subdigraph on $B$:
\begin{center}
	\includegraphics[scale=1.2]{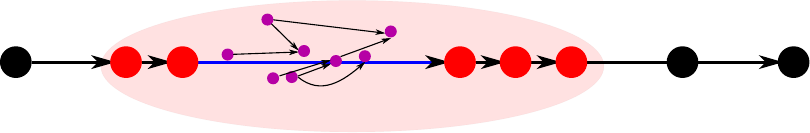}
\end{center}

Recall that running Algorithm~\ref{algo:lspdag} with a pseudo-topological order is equivalent to finding the longest path on the digraph that results by erasing all the edges that go backward. Therefore, we must consider \emph{only} orders of $B$ that are themselves pseudo-topological orders of the induced subgraph on $B$.

We describe three approaches to choosing an order of $B$.

\subsubsection{The powerful approach}\label{powerful}
The ``powerful'' approach is to recursively repeat the process on the induced subgraph on $B$. That is, repeat for $B$ the whole process of finding the strongly connected components, performing DFS as described in a previous section, finding suitable pseudo-topological orders, opening their respective edges, and recursively doing the same for the induced subgraphs in order to find good pseudo-topological orders and so on.

The problem with this approach is that with any standard cache-friendly graph data structure we would need to reconstruct the induced digraph (\emph{i.e.} rename the vertices of $B$ but keep track of their original names), and the whole process is slow. Of course, we would only need to do this process once per connected component and then we can use the results for each edge of the path.

The advantage of this approach is that we are precisely recursively finding good pseudo-topological orders of $B$, which means it's likely many long paths can be inserted in our original $P$.

\subsubsection{The heuristic approach}\label{heuristic}
Instead of attempting to repeat the whole algorithm on the induced subgraph, we try to mimic heuristically its results. Consider the following operation on the inbetween vertices: pick a vertex $u$ at random, and exchange its position with some out-neighbor $v$ of $u$ which appeared before $u$ in the pseudo-topological order. If no such neighbor exists, simply find pick another $u$.

Repeat this operation (which is quite inexpensive) as many times as the time limit allows. The following theorem ensures that this process will likely end in a (weak) pseudo-topological order.

\begin{theorem}
	With probability approaching 1 as the number of times approaches infinity, the order constructed above is a weak pseudo-topological order of the induced subgraph $B$.
\end{theorem}

\textbf{Proof:} For any digraph, given a total order of the vertices, call a pair of vertices $(a,b)$ \emph{bad} if $a$ appears before $b$ in the order, there is a path from $b$ to $a$ but not one from $a$ to $b$. In other words, if the strongly connected component which contains $b$ is (strictly) less than the strongly connected component which contains $a$ in the partial order of the skeleton of $D$.

Thus, we only need to prove that the number of bad pairs never increases after an operation, and that with some positive probability, it decreases.

Suppose we do an operation, and $u$ is the chosen vertex, which is exchanged with $v$ (so $u$ was after $v$ in the order, but there is an edge from $u$ to $v$). Then only pairs of vertices of the form $(a,u)$ and $(v,a)$ could have changed status (and only when $a$ is between $u$ and $v$ in the order).

Let $U,V,A$ be the strongly connected components containing $u,v,a$ respectively. If $A$ is before $U$, then indeed the pair $(a,u)$ is now bad after the exchange, but since we are assuming there is an edge from $u$ to $v$, if $A$ is before $U$, then it is before $V$, and so the pair $(v,a)$ used to be bad, but is now good. When $(v,a)$ becomes bad, the process is analogous, and $(a,u)$ becomes good. \Qed

The above process then gives a random approximate algorithm to calculate weak pseudo-topological orders without calculating the strongly connected components.

\subsubsection{The theoretical approach}\label{theoretical}

Denote $E(P)$ the edge set of a path $P$.

For a positive integer $k$, we wish to find $P$ that is maximal in the sense that if $Q$ is another path for which $|E(Q) \setminus E(P)| \leq k$, then the total weight of $Q$ is less than or equal to the total weight of $P$.

Given the edge opening process, this problem can be reduced to the following: we wish for a minimal set of permutations of $S_n$ for which every $k$-subset of $n$ appears in every possible order. The idea is to try an edge opening for every edge in the path with the order of $B$ given by each element of the set of permutations.

This problem has been worked on by Colbourn \emph{et al.} on \cite{chee2013sequence} and \cite{murray2014sequence}, where they named any such set of permutations a \emph{sequence covering array}. They give a mostly impractical algorithm for finding covering arrays, that works in practice up to $n\approx 100$. However, an easy probabilistic argument yields that taking $\Theta(\log(n))$ permutations randomly gives a non-zero chance of ending up with a covering array. This suggests that merely taking many random permutations would yield (probabilistically) the desired result. Unsurprisingly, this approach is not nearly as efficient as the other two.

For $k = 2$, however, a covering array is easy to find: take any permutation and its reverse. So by opening every edge and taking any permutation of the inbetween vertices and its reverse, we ensure the found path is optimal in this very limited sense: There exists no other path with higher total weight all whose edges, except one, are a subset of the edges of $P$ plus one more.

\section{Some details about the implementation}\label{sec:impl}

\subsection{Preprocessing the graph}

The data structure we use for storing the graph is a simple pair of adjacency lists: one for the out-neighbors and one for the in-neighbors, so we can access either efficiently. The vertices are numbered from $0$ to $n-1$ and we store an array of size $n$, each element of which is an array storing both the neighbors of the corresponding vertex and the corresponding edge-weights.

Next, we find connected components. While it is true that finding the weakly connected components, strongly connected components and skeleton might require some non-trivial processing, this pales in comparison to the NP-hard problem of finding the longest path. An efficient implementation (for example, using C++ boost graph library) can find the weakly and strongly connected components on graphs with millions of edges in a fraction of a second.

Then, we find the out-heuristic and the in-heuristic scores for each vertex, as described in Section~\ref{DFS}, and sort the neighbors of each vertex according to the heuristic.

In our experiments, the whole preprocessing step took about 0.2 seconds on the Oracle graph, which has $\sim$13,300 vertices. This time includes reading the file of strings, figuring out which concatenations are possible and constructing the graph. Experiments with randomly generated graphs of comparable size take less than 0.1 seconds if the graph is already in memory.

If one has a training set of a class of graphs, one could use some rudimentary machine learning to try to find the optimal parameters so that on average good paths are found. In fact, for the contest, we did just that, which provided a slight boost. The code includes a trainer for this purpose, but experimental results on the benefits of this are sketchy at best and do not (usually) warrant the long time it takes to train.

\subsection{Pseudo-Topological orders}

Once we have a pseudo-topological order $T$, we construct its inverse for fast access, so in addition of being able to answer the query ``which vertex is in position $i$ of $T$?'' we can also answer the query ``at which position is vertex $v$ in $T$?'' efficiently. Therefore, any operation we do on $T$ must be mirrored to the inverse.

In addition, since we are constantly changing $T$ and having to rerun Algorithm~\ref{algo:lspdag}, it is worth it to store $x_v$ for each $v$, and just reprocess from the first vertex whose order was modified and onwards.

Fortunately, when performing the edge opening process, much of the order has been preserved, so we can use this information and recalculate from the first modification onwards, speeding up the calculation considerably.

Finally, opening the edges is just a matter of rearranging the vertices to satisfy the condition, which is straightforward. We found experimentally that the process finds good paths faster if the edges of the path are opened in a random order and not sequentially. This makes intuitive sense. If a path cannot be improved by a certain edge opening, it's unlikely (but possible) an edge that is near will yield an improvement.

Our implementation of the ``powerful'' approach described in Section~\ref{powerful} was by constructing a completely new graph and running the algorithm on the subgraph recursively, and so it was prohibitely slow, although it did tend to find longer subpath insertions with fewer edge openings. Perhaps this can be improved. The implementation of the heuristic approach of Section~\ref{heuristic} was considerably more efficient over the random approach described in Section~\ref{theoretical}.

\section{Experimental data}\label{experiments}

We compare this algorithm to one other found in the literature, by Portugal \textit{et. al.} \cite{portugal2010study}, as the authors have kindly provided us with the code. There is a scarcity of heuristic algorithms for this problem, and the code for some, such as \cite{scholvin1999approximating} appears to have been lost, so a direct comparison turns out to be impractical. Unfortunately, an extensive literature search did not provide any other accessible source code for this problem, making the code in the link of the introduction the only open source and readily available implementation we are aware of that heuristically finds long simple paths. 

In \cite{portugal2010study}, the authors compare four approaches based on genetic algorithms. The biggest graph they used as an example consists of 134 vertices. The result was that their fastest algorithm was able to find the optimal solution more than half of the time in around 22 seconds. For comparison, our program only took 0.001 seconds for the same graph and found the longest path on 100\% of the test runs. Please bear in mind that the comparison might be somewhat unfair, since their implementation was in Matlab instead of C++. Our algorithm took less than a millisecond for all the graphs in \cite{portugal2010study} and found the longest simple path 100\% of the test runs.

\subsection{Tests in large random graphs where we know the longest path size}

Given $n$ and $m$, consider the following graph generation process for a graph with $n$ vertices and $m$ edges.

Consider any random permutation of the vertices $v_1, v_2, ..., v_n$ and add all edges $v_i \to v_{i+1}$. Then pick $n-m+1$ other edges uniformly at random. All edge weights were set to 1, so we know for certain the longest simple path has size $n-1$.

For example, in our experiments, for $n = 10,000$ and $m = 100,000$, the whole process (including reading the file containing the graph) took on average 1.28 seconds to find the longest simple path.


\vspace{0.3cm}

\section{Acknowledgements}

We would like to thank the organizers of the Oracle MDC coding challenge for providing a very interesting problem to work on (and for the prize of a drone and camera, of course). Furthermore, we would like to thank the other participants, specially Miguel Ángel Sánchez Pérez and David Felipe Castillo Velázquez for the fierce competition. Also, we are grateful to Marisol Flores and Edgardo Roldán for their helpful comments on the paper, as well as David Portugal for providing the source code from their work. This research was partially supported by PAPIIT IA106316.

\vspace{-0.3cm}

\providecommand{\bysame}{\leavevmode\hbox to3em{\hrulefill}\thinspace}
\providecommand{\MR}{\relax\ifhmode\unskip\space\fi MR }
\providecommand{\MRhref}[2]{%
\href{http://www.ams.org/mathscinet-getitem?mr=#1}{#2}
}
\providecommand{\href}[2]{#2}

\end{document}